\documentclass[]{llncs}

\usepackage[titletoc]{appendix}

\usepackage {amssymb}
\usepackage {amsmath}
\usepackage [usenames] {xcolor}
\usepackage {enumerate}
\usepackage[mathlines]{lineno}
\usepackage{graphicx}  
\usepackage{hyperref}
\usepackage{algorithm}
\usepackage{algorithmic}
\usepackage{subcaption}
\usepackage{xspace}
\usepackage{a4wide}

\newcommand{\etal}{{\it et al.\xspace}}

\def\C{{\cal C}}

\def\N{{\cal N}}

\def\V{{\cal V}}

\newcommand{\Hausdorff}{\mathsf{H}}
\newcommand{\diHausdorff}{\overrightarrow{\mathsf{H}}}

\newcommand{\hamcycle}{\textsc{Hamiltonian Cycle}\xspace}
\newcommand{\hampath}{\textsc{Hamiltonian Path}\xspace}
\newcommand{\curvesimp}{\textsc{Directed Curve Simplification}\xspace}
\newcommand{\segcover}{\textsc{Segment Polyline Cover}\xspace}
\newcommand{\consegcover}{\textsc{Connected Segment Polyline Cover}\xspace}

\newcommand{\Reals}{{\mathbb{R}}}            

\newcommand {\tablecell}[1]
{
  \begin {minipage} {3.5cm}
  \centering
  #1
  \end {minipage}
}
\newcommand {\tablerow}[4]
{
  #1
  & \tablecell{#2}
  & \tablecell{#3}
  & \tablecell{#4}
  \\
}

\newcommand{\polylog}{\mathrm{polylog}}
\newcommand{\poly}{\mathrm{poly}}

\newtheorem{observation}{Observation}

\title {Embedding Ray Intersection Graphs \\ and Global Curve Simplification}

\author
{
  Mees van de Kerkhof$\dagger$\inst{1}
  \and
  Irina Kostitsyna\inst{2}\orcidID{0000-0003-0544-2257}
  \and
  Maarten L\"offler\inst{1}
}
\authorrunning{M. van de Kerkhof, I. Kostitsyna and M. L\"offler}

\institute
{
  Department of Computing and Information Sciences, Utrecht University\\
  \email{m.vandekerkhof@uu.nl \quad m.loffler@uu.nl}
  \and
  Department of Mathematics and Computer Science, TU Eindhoven\\
  \email{i.kostitsyna@tue.nl} \\
  $\dagger$: Corresponding author
}

\begin{document}
\maketitle
\begin{abstract}
  We prove that circle graphs (intersection graphs of circle chords) can be embedded as intersection graphs of rays in the plane with polynomial-size bit complexity.
  
  We use this embedding to show that the global curve simplification problem for the directed Hausdorff distance is NP-hard.
  In this problem, we are given a polygonal curve $P$ and the goal is to find a second polygonal curve $P'$ such that the directed Hausdorff distance from $P'$ to $P$ is at most a given constant, and the complexity of $P'$ is as small as possible.
  

\end{abstract}




\section {Introduction}

Problems in the area of graph drawing often find application in complexity theory by providing a basis for NP-hardness proofs for geometric problems. 
%
In this paper, we study an application of embedding {\em circle graphs} (intersection graphs of chords of a circle) as {\em ray graphs} (intersection graphs of half-lines) to the analysis of the complexity of {\em global curve simplification}.
In particular, we prove (refer to Section~\ref {sec:pre} for precise problem definitions):

\begin {itemize}
  \item All circle graphs are ray graphs that have a representation as a set of intersecting rays described by coordinates that have a polynomial number of bits   (Theorem~\ref {thm:circleray});
  \item \hampath is NP-hard on such ray graphs (Corollary~\ref {cor:hamraypolyhard});
  \item \curvesimp is NP-hard (Theorem~\ref {thm:finalcellhard}).
\end {itemize}


\subsection {Global curve simplification}
Curve simplification is a long-studied problem in computational geometry and has applications in many related disciplines, such as graphics, and geographical information science (GIS).
Given a polygonal curve $P$ with $n$ vertices, the goal is to find another polygonal curve $P'$ with a smaller number of vertices such that $P'$ is sufficiently similar to $P$.
Methods proposed for this problem famously include a simple heuristic scheme by Douglas and Peucker~\cite{douglas73algorithms}, and a more involved classical algorithm by Imai and Iri~\cite{imai88algorithms}; both are frequently implemented and cited.
Since then, numerous further results on curve simplification, often in specific settings or under additional constraints, have been obtained~\cite{abam10streaming,ahmw-nltaa-05,barequet2002approx3,berg98correct,buzer07optimal,cc-apcwmnls-96,chen05angle,g-nmfc-96,ii-aoaaplf-86}.

Recently, the distinction was made between {\em global} simplification, when a bound on a distance measure must be satisfied between $P$ and $P'$, and {\em local} simplification when a bound on a distance measure must be satisfied between each edge of $P'$ and its corresponding section of $P$~\cite{kklmw-gcs-19}.
A local simplification is also a global simplification, but the reverse is not necessarily true, see Figure~\ref {fig:globallocal}.
\begin{figure} [t]
	\centering \includegraphics[width=\textwidth] {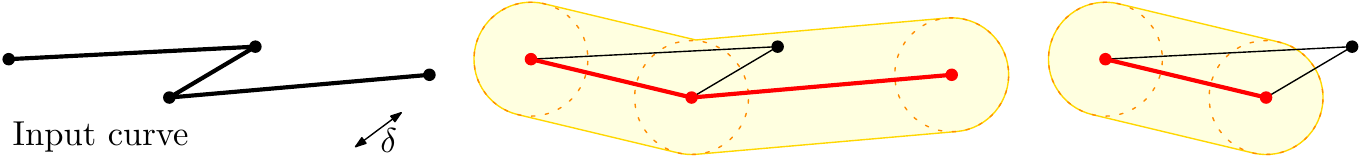} 
\caption{For a target Hausdorff distance $\delta$, the red curve (middle) is a global simplification of the input curve (left), but it is not a local simplification, since the first shortcut does not closely represent its corresponding curve section (right).}
	\label{fig:globallocal}
\end{figure}

Agarwal~\etal~\cite {ahmw-nltaa-05} were first to consider the idea of global simplification under the {\em Fr\'echet distance}. They introduce what they call a {\em weak simplification}: a model in which the vertices of the simplification are not restricted to be a subset of the input vertices, but can lie anywhere in the ambient space. 
Kostitsyna~\etal~\cite{klps-ocmlpp-17} present a polynomial-time algorithm for this model but for the {\em Hausdorff distance}; in particular, the directed  Hausdorff distance from the simplification curve to input curve.
Van Kreveld~\etal~\cite{klw-oopsihfd-18} consider a different setting in which the output vertices should be a subsequence of the input, and they also consider the Hausdorff distance. They give a polynomial-time algorithm for the directed  Hausdorff distance from the simplification curve to input curve, but they show the problem is NP-hard for the directed Hausdorff distance in the opposite direction, and also for the symmetric (undirected) Hausdorff distance.
Van de Kerkhof \etal~\cite {kklmw-gcs-19} prove that the hardness result for the unrestricted Hausdorff distance can be extended to the non-restricted case as well; in addition, they introduce an intermediate {\em curve-restricted} model where the vertices of the simplified curve should lie on the input curve. Surprisingly, the problem is hard under this model for all three variants of the Hausdorff distance.
Table~\ref {tab:results} summarizes the state of the art for global curve simplification under the Hausdorff distance.

\begin{table}[!ht]
\centering 
\caption{Results for global curve simplification under the Hausdorff distance between the curve $P$ and its simplification $P'$. The result in {\bf bold} is from this work.}
\begin{tabular}[c]{|c|c|c|c|}
\hline 
Distance & Vertex-restricted $(\V)$ & Curve-restricted $(\C)$ & Non-restricted ($\N$)\\ 
\hline \hline
\tablerow{$\diHausdorff(P,P')$}
{ 
  NP-hard~\cite{klw-oopsihfd-18}
}
{ 
  NP-hard 
  \cite{kklmw-gcs-19}
}
{ 
  \bf NP-hard 
}
\hline
\tablerow{$\overrightarrow{H}(P',P)$}
{ 
  $O(n^4)$~\cite{klw-oopsihfd-18} \\
  O($n^2\polylog\ n$)
  \cite{kklmw-gcs-19}
}
{ 
  NP-hard 
  \cite{kklmw-gcs-19}
}
{ 
  $\poly(n)$ \cite{klps-ocmlpp-17}
}
\hline
\tablerow{$\Hausdorff(P,P')$}
{ 
  NP-hard~\cite{klw-oopsihfd-18}
}
{ 
  NP-hard
  \cite{kklmw-gcs-19}
}
{ 
  NP-hard 
  \cite{kklmw-gcs-19}
}
\hline
\end{tabular}

\label{tab:results}
\end{table}

\subsection {Embeddings of geometric intersection graphs}

Geometric intersection graphs have long been studied due to their wide range of applications, and lie on the interface between computational geometry, graph theory, and graph drawing
\cite {mc2}.
The graph classes corresponding to intersections of geometric shapes form a natural hierarchy that links to the complexity of those shapes: more complex shapes allow to represent more graphs.
Arguably the most restricted class in this family are the {\em unit interval graphs}~\cite {yuanzhou}, 
and the most general class of intersection graphs of connected shapes in $\Reals^2$ are the {\em string graphs}~\cite {zbMATH03542461}.
Between these two, a hierarchy of classes exist; part of it is illustrated in Figure~\ref {fig:family}.

\begin{figure} [t]
	\centering \includegraphics[] {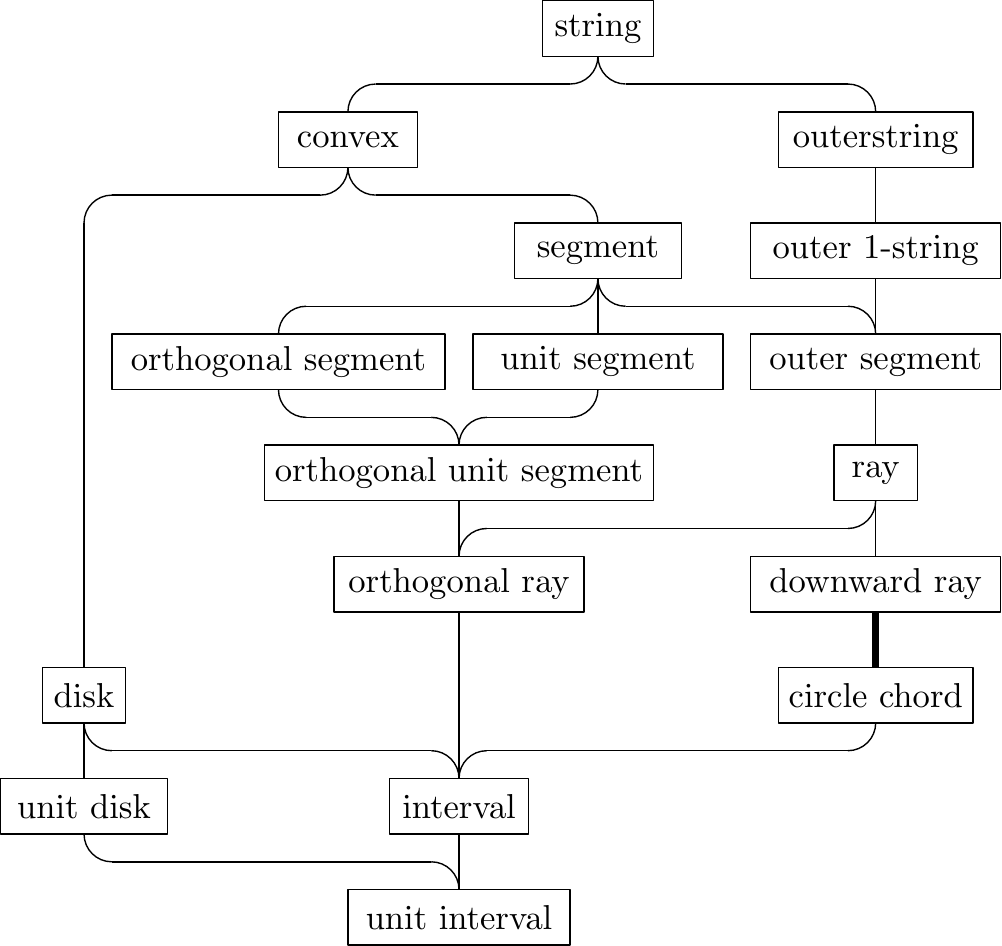} 
\caption{Intersection graph classes and their inclusion relations. The thickened edge indicates our contribution. In order to keep the figure readable many classes and refinements have been omitted; for an extensive overview we refer the reader to e.g.~\cite {11303_4499,cj-rhcgig-17,JGAA-470} or to 
\url{graphclasses.org}.}
	\label{fig:family}
\end{figure}

Hartmann~\etal~\cite {hart} introduce {\em grid intersection graphs} where the shapes are aligned to an orthogonal grid; Mustata~\cite {11303_4499} gives an overview of the state of the art and also discusses the complexity of computational problems on such classes.
Cardinal~\etal~\cite {JGAA-470} prove several relations between segment intersection graphs and ray intersection graphs; in particular they introduce {\em downward ray graphs}: intersection graphs of rays that all point into a common half-plane. Their main result is that {\em recognition} of several classes is complete for the existential theory of the reals.
{\em Circle graphs} are intersection graphs of chords of a circle; equivalently, they may be defined as interval graphs where there is an edge between two intervals on a line when they intersect but are not nested~\cite {zbMATH03859178}.
Circle graphs are known to be contained in 1-string graphs~\cite {10.5555/1283383.1283449}. We are not aware of any published statements of stricter containment; we show in this paper that they are in fact contained in the downward ray graphs.


When utilizing graph embedding algorithms in hardness proofs, one important issue is the representation of the embedding. 
The class of graphs which can be represented as intersection graphs of a given set of shapes is not necessarily the same as the class of graphs which can be represented as intersection graphs of shapes which each can be represented with coordinates of bounded complexity.
      For instance, McDiarmid and M{\"uller}~\cite {mdm-irdsg-13} show that not all realizable unit disk graphs can be realized with coordinates of logarithmic complexity, and the same is true for {\em segment} graphs \cite{km-igs-94}.

\section {Preliminaries, Overview \& Challenges} \label {sec:pre}

\subsection{Polygonal curves and the Hausdorff distance}
A \emph{polygonal curve} (also called a \emph{polyline})  $P = \{p_1,p_2,\ldots,p_n\}$ is defined by an ordered sequence of $n$ vertices. We can treat $P$ as a continuous map $P : [1,n] \rightarrow \mathbb{R}^d$ that maps real values in the interval \([1 \ldots n]\) to points
on the polyline by linearly interpolating between the vertices, which allows us to visualize a polyline as $n-1$ line segments linked one after the other. We will refer to these segments as the polyline's \emph{links}. For integer $i$, $P(i)$ will return the vertex $p_i$. Points on $P$'s $i$'th link are parametrized as
\( P(i + \lambda) = (1-\lambda)p_i + \lambda p_{i+1} \).
The \emph{directed Hausdorff distance} from curve $P$ to curve $Q$, which have $n$ and $m$ vertices respectively, is given by $\overrightarrow{H}(P,Q) = \displaystyle\max_{i \in [1 \ldots n]}\min_{j \in [1 \ldots m]} ||P(i) - Q(j)||$. I.e. it is equal to the Euclidean distance from the point on $P$ furthest from $Q$ to the point on $Q$ closest to that point. The \emph{(undirected) Hausdorff distance} is the maximum over both directions, i.e. $H(P,Q) = \max \{\overrightarrow{H}(P,Q),\overrightarrow{H}(Q,P)\}$.

\subsection{Problem}
The problem we wish to tackle (and which we will prove NP-hard) is:
\begin {problem} \curvesimp.
  Given a polyline $P$, integer $k$ and a value $\delta$,
  find another polyline $P'$
  such that the directed Hausdorff distance from $P$ to $P'$ is at most $\delta$
  and the number of links in $P'$ is at most $k$.
\end {problem}
Note that van de Kerkhof~\etal~\cite{kklmw-gcs-19} call this problem the {\em 
non-restricted global curve simplification problem} to distinguish it from other variants; in the remainder of the present paper we use the shorter name for convenience.
We will find that the key difficulty in solving \curvesimp lies in the following similar problem:
\begin {problem} \segcover.
  Given a set $L$ of line segments in the plane and integer $k$,
  find a polyline $P$
  such that every segment in $L$ is {\em covered} by $P$ (contained in at least one segment of $P$),
  and $P$ has at most $k$ links.
\end {problem}

\subsection {Proof idea}

Our approach is to show that \segcover is hard by a reduction from \hampath on ray intersection graphs.
Specifically, we use the following idea.

\begin{figure} [t]
	\centering \includegraphics [width=\textwidth] {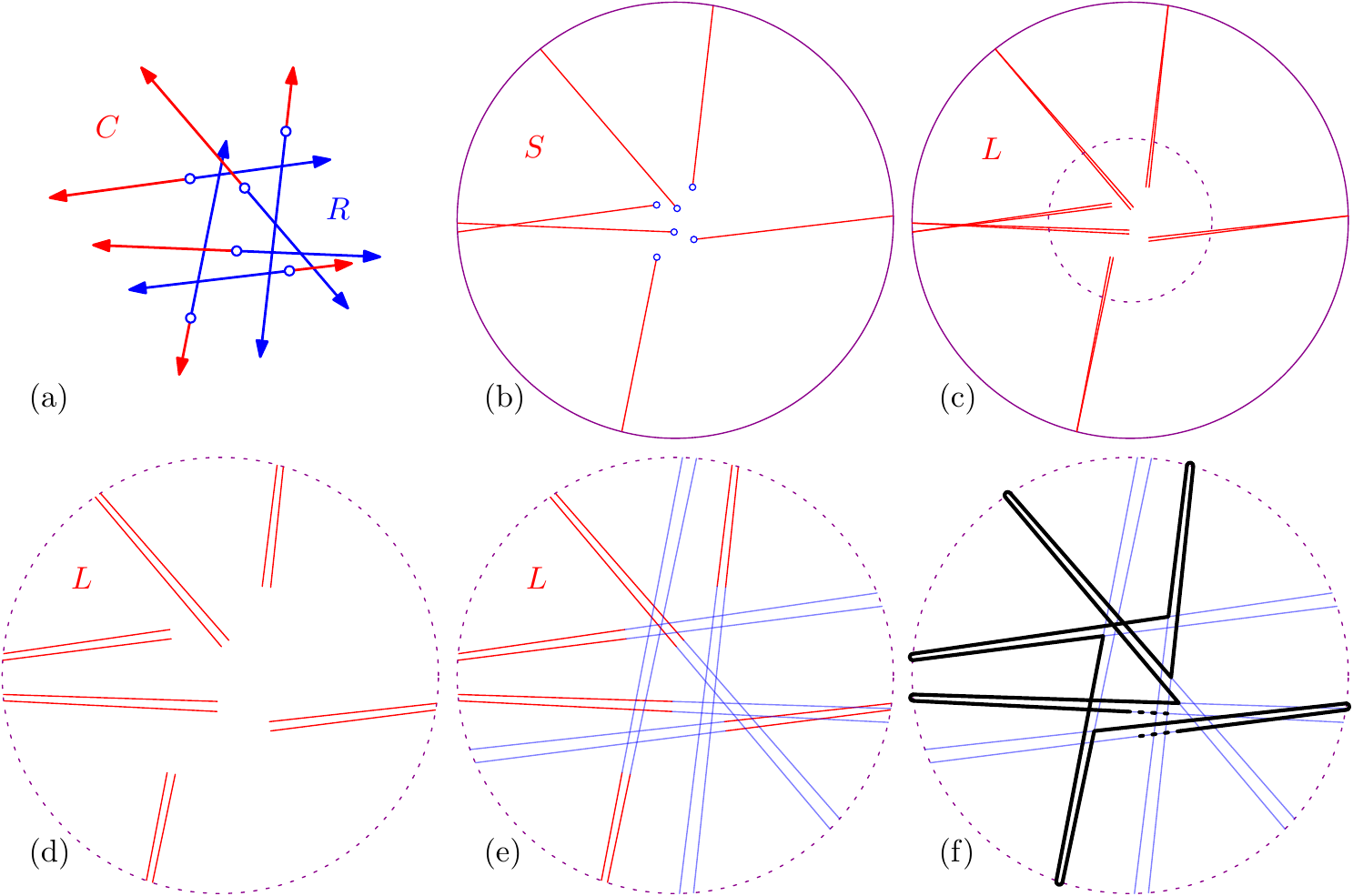} 
    \caption
    { The idea for a small example (which does not admit a Hamiltonian cycle).
    (a) A set of rays $R$ (blue) whose intersection graph is $G$, and the complement $C$ (red).
    (b) Zooming out until we can draw a circle that contains all intersections among rays in $C$.
    (c) Replacing each ray in $C$ by a {\em needle}.
    (d) Zooming back in.
    (e) The extensions of the needles (blue) correspond to the original rays.
    (f) A polygon covering all needles must correspond to a Hamiltonian cycle in $G$ (here, there is no solution).
    }
	\label{fig:idea}
\end{figure}

\begin {observation}
\label {obs:key}
  Let $G$ be a ray intersection graph with $n$ vertices.
  There exists a set $L$ of $2n$ segments
  such that $G$ has a Hamiltonian cycle if and only if
  there is a polygon covering $L$ with $2n$ vertices.
\end {observation}
We can use Observation~\ref{obs:key} to prove \textsc{Segment Polyline Cover} is NP-hard and then reduce \textsc{Segment Polyline Cover} to \textsc{Directed Curve Simplification}, proving it NP-hard as well.

We sketch the proof of Observation~\ref {obs:key} here; the rest of the paper is devoted to making it precise.
The high level proof idea is illustrated in Figure~\ref {fig:idea}.
Let $R$ be a set of rays in $\Reals^2$, and $G$ its intersection graph.
The {\em complement} of a ray $r$ is the ray with the same origin and the same supporting line as $r$ which points in the opposite direction.
Let $C$ be the complement of $R$.
We cut the rays in $C$ to a set of segments $S$ in such a way that $C$ and $S$ have the same intersection graph.
Then we replace each segment $s \in S$ by a {\em needle}: a pair of segments both very close to $s$ that share one endpoint (different from the corresponding ray's origin).
Let $L$ be the resulting set of $2n$ segments.
Now, any polygon with $2n$ segments covering $L$ must use the two edges of one needle consecutively (since, by construction, the extension of these segments does not intersect the supporting line of any other segment), and it can connect an edge from one needle to an edge of another needle exactly when the corresponding original rays in $R$ intersect.

\subsection {Challenges}

Though the idea is conceptually simple,
there are several difficulties in turning Observation~\ref {obs:key} into a proof that \curvesimp is NP-hard.
\begin {itemize}
  \item The simple idea above is phrased in terms of a \hamcycle and covering segments by a polygon; for our proof we need to use a polyline. We need to be careful in how to handle the endpoints.
  \item We need to establish that \hampath is indeed NP-hard on ray intersection graphs.
  \item We need to know how to embed a ray intersection graph as an actual set of rays with limited bit complexity.
  \item We need to model the input to \curvesimp as an instance of \segcover. Specifically, the complement of a set of rays is not necessarily connected; but the input to \curvesimp must be connected.
  \item The \segcover problem closely resembles \curvesimp for $\delta = 0$; to extend it to the case $\delta > 0$ we (again) need to carefully consider the complexity of the embedding.
\end {itemize}
Most of these challenges can be overcome, as we show in the remainder of this paper. 
However, since the problem of recognizing if a graph can be embedded as a set of intersecting rays is complete for the existential theory of the reals~\cite {JGAA-470}, we know that there are ray intersection graphs that cannot be embedded by a set of rays with subexponential bit complexity, unless $\textbf{NP} = \exists\mathbb{R}$.
In this paper, we work around  this problem by considering a smaller class of graphs, and allowing a superpolynomial grid for our embeddings, which we show is sufficient for the proof of Theorem~\ref {thm:finalcellhard}.

\section {Hamiltonian cycles in ray intersection graphs}
\label {sec:hamray}
 
  \subsection {Embedding circle graphs as ray graphs}

We will show that each circle graph can be embedded as a ray intersection graph. To show this, we construct a set of \(n\) points that lie on a convex, increasing curve such that all chords connecting a pair of points can be extended to a ray to the right, and none of these rays will intersect below the curve. This requires the curve to grow very fast. We use the points \((x,x!)\) for \(x \in [1 .. n]\), where $x! = \Pi_{i=1}^x i$ is the factorial function. 
Indeed, these points have the following property.

\begin {lemma}
\label{lem:case1}
  Let $a, b, c, d \in [1 .. n]$ be four numbers such that $a < b$ and $c < d$. 
  Let $A$ be the ray starting at $(a, a!)$ and containing $(b, b!)$,
  and let $B$ be the ray starting at $(b, b!)$ such that $B \subset A$.
  Similarly, let $C$ be the ray starting at $(c, c!)$ and containing $(d, d!)$,
  and let $D$ be the ray starting at $(d, d!)$ such that $D \subset C$.
  Then $B$ and $D$ do not intersect; hence, $A$ and $C$ intersect if and only if $A \setminus B$ and $C \setminus D$ intersect.
\end {lemma}

\begin {proof}
Since every ray is drawn between two points on the curve of the function $x!$, we know that it intersects this curve only at these points.
The distance between $y$-coordinates of successive points keeps rapidly increasing as $x$ increases, but the distance between $x$-coordinates of successive points is constant.
Thus the slope of a ray $r_1$ whose intersection points with the curve lie to the right of those of ray $r_2$ will be greater than the slope of $r_2$.
Without loss of generality we assume $a < c$. There are three possible cases, see Figure~\ref{fig:casedistinction}:
  \begin{figure}
      \centering
      \includegraphics[width=0.85\textwidth]{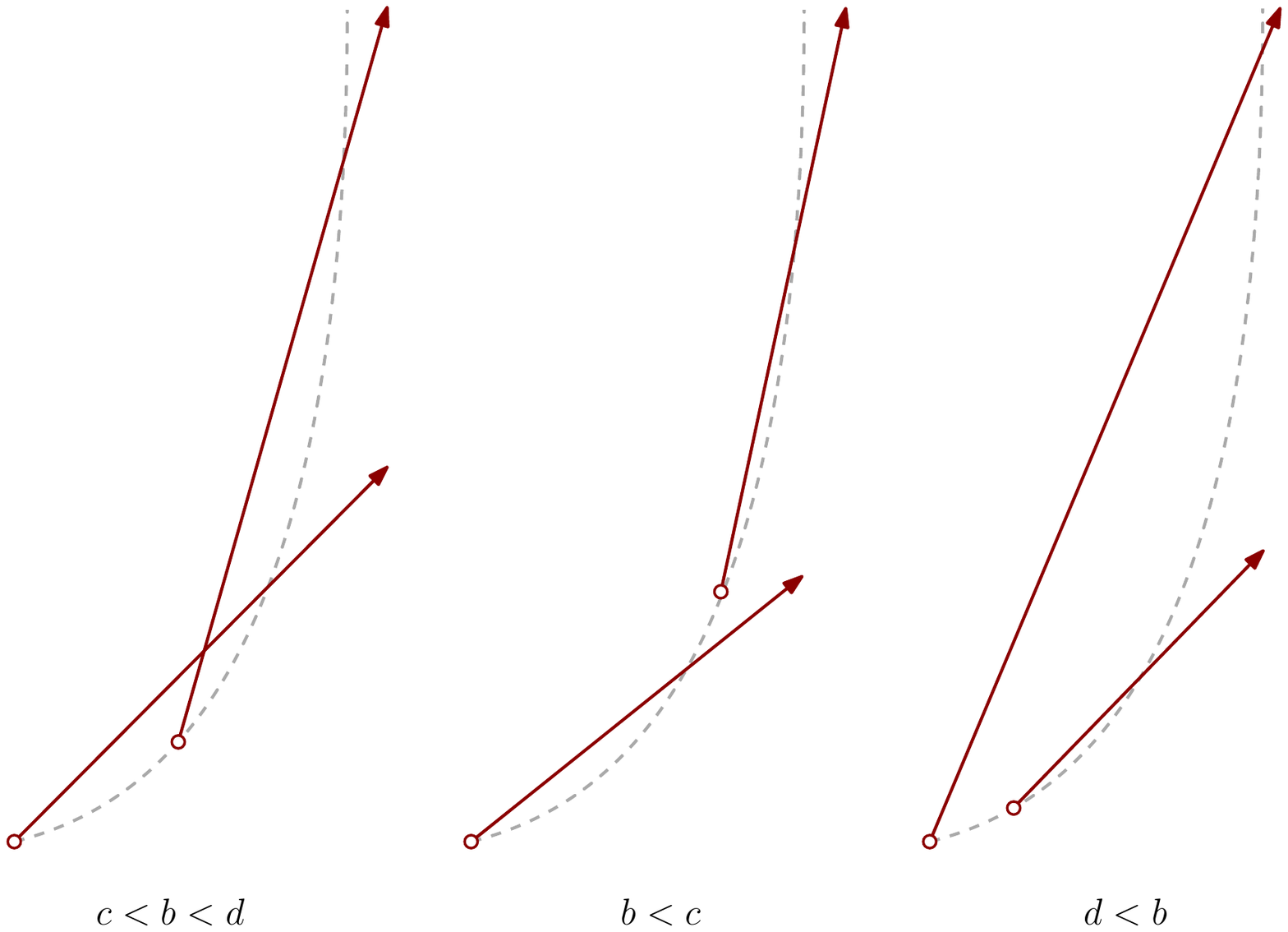}
      \caption{The three cases for two rays. $a$ and $b$ are the $x$-coordinates of the points where the first ray intersects with the curve $y=x!$, and $c$ and $d$ are those values for the second ray. No matter the case, the rays will not intersect below the curve. 
      }
      \label{fig:casedistinction}
  \end{figure}
  \begin{itemize}
   \item $c < b < d$: Here it is clear that $A$ and $C$ will intersect at an $x$-coordinate somewhere between $c$ and $b$, and so $B$ and $D$ will not intersect.
  \item $b < c$: Here we can easily see that $B$ and $D$ do not intersect, as $B$ starts below $D$ and has a lower slope.
  
  \item $d < b$: Whereas the first two cases only require the curve to be convex and increasing, this case also requires the function to grow quick enough: Since $D$ starts to the left of $B$ it could possibly intersect $B$ if its slope was higher. We will now show, however, that the factorial function grows quick enough so that this cannot happen.
  For a fixed $b$, the lowest slope that $B$ can have is when $a=1$. The highest slope that $D$ can have occurs when $c=b-2$ and $d=b-1$. The slope of $B$ is equal to the slope of $A$, which would be $\frac{b! - 1!}{b-1} = b\cdot (b-2)! - \frac{1}{b-1}$. The slope of $D$ (and $C$) in this scenario would be $\frac{(b-1)! - (b-2)!}{1} = (b-2) \cdot (b-2)!$. We can see that the slope of $B$ is higher than $D$ if $2\cdot(b-2)! \geq \frac{1}{b-1}$ which obviously holds for all $b>2$. So since $B$ starts above $D$ and has higher slope, $B$ and $D$ will not intersect.
  \end{itemize}
\end{proof}
Once we have constructed these points we can ``unroll'' any circle graph by picking one chord endpoint on the circle to be the first point and then traversing the circle in clockwise order and assigning each chord endpoint we encounter the next point of our set. See Figure~\ref{fig:chordraysketch} for a sketch. Because the $y$-coordinate for a point will not grow bigger than $O(n^n)$ we can represent the points using polynomial bit complexity.
\begin{figure}[t]
\centering
\includegraphics{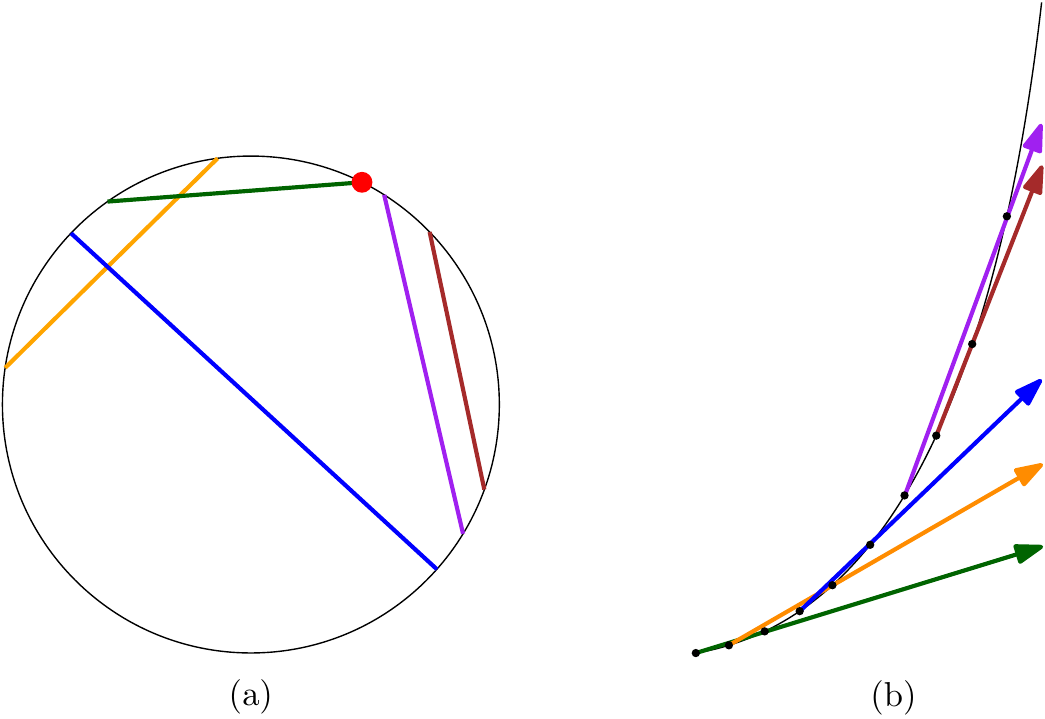}
\caption{(a) A circle graph with colors assigned to the chords. The chosen starting point is marked in red. (b) Unrolled version of (a), by assigning chord endpoints to points on the convex curve they can be extended into rays without intersecting.}
\label{fig:chordraysketch}
\end{figure}
  At this point, we have shown that circle graphs are contained in ray graphs. In fact, our construction gives a bit more:

\begin {theorem} \label {thm:circleray}
  The class of circle graphs is contained in the class of ray intersection graphs.
  Furthermore, every circle graph can be embedded as the intersection graph of a set of rays such that:
  \begin {itemize}
  \item every ray is grounded on a common curve {\em(grounded ray graph~\cite {JGAA-470})};
  \item every ray points towards the upper right quadrant {\em(downward ray graph~\cite {JGAA-470})};
  \item every ray is described by a point and a vector with polynomial bit complexity.
  \end {itemize}
\end {theorem}
%

\subsection {Hamiltonian paths and cycles}

Next, we show that Hamiltonian Path problem is NP-hard on ray graphs, and in particular, on ray graphs with polynomial bit complexity.
 
We reduce from the \textsc{Hamiltonian Path} problem on circle graphs. We make use of the proof from Damaschke~\cite{DAMASCHKE19891}. He shows that Hamiltonian cycle is NP-hard on circle graphs, by reducing from Hamiltonian Cycle in cubic bipartite graphs. He also claims that there is an easy adaptation that shows the Hamiltonian path problem is also NP-hard for circle graphs. We will start by making this adaptation explicit: We construct an instance of the circle graph problem as described in~\cite{DAMASCHKE19891}, but then we replace one of the $X$-chords with two parallel chords close to where the $X$-chord was, so that they both intersect the same chords that were intersected by the $X$-chord. For both of the new chords we then add one new chord that only intersects that chord and no others. Now we know that the circle graph will have a Hamiltonian path if and only if the bipartite graph has a Hamiltonian cycle. 
From Theorem~\ref {thm:circleray} we now immediately have:

%

\begin {corollary} \label {cor:hamraypolyhard}
  Hamiltonian Path is NP-hard on intersection graphs of rays that have a polynomial bit complexity.
\end {corollary}

\section {\consegcover}

Next, we introduce the \textsc{Connected Segment Polyline Cover} problem,
and show that it is NP-hard by a reduction from Hamiltonian Path problem on circle graphs through the construction outlined above.

\begin {problem} \consegcover.
  Given a set $L$ of $n$ line segments whose union is connected, and an integer $k$,
  decide if there exists a polyline of $k$ links that fully covers all segments in $L$.
\end {problem}

We start by embedding the circle graph as a ray intersection graph in the manner outlined above. Then, we compute all intersection points between supporting lines of the rays. One of these intersection points will have the lowest \(y\)-coordinate. We will then choose a value that is lower than this lowest $y$-coordinate, which we will denote as $y_\ell$. For each ray $r$, let $p_r$ be its starting point. Let $\bar r$ be $r$'s complement: the part of the supporting line that is not covered by $r$. Let $\tilde r$ be the part of $\bar r$ that has $y \geq y_\ell$. Now we construct a \emph{needle} for each ray's complement: Two line segments that share one endpoint at the point where $\bar r$ has $y$-coordinate $y_\ell$. The other endpoint for both segments lies very close to $p_r$. The endpoints are on opposite sides of the ray starting point so we get a wedge-like shape that runs nearly parallel to $\tilde r$. In addition to these $2n$ segments, which we will refer to as \emph{needle segments}, we create three more segments which we will refer to as the \emph{leading segments}: We create one horizontal segment we call $s_h$ with $y$-coordinate between $y_\ell$ and the lowest intersection point between ray supporting lines, that starts far to the right of the needle segments and ends to the left of them, intersecting all of the needles. Attached to $s_h$ is a large vertical segment we call $s_v$, running up to a point above the highest starting point of a ray. Attached to that is another horizontal segment we call $s_t$, this one being short and ending to the left of any ray starting point. See Figure~\ref{fig:consegpolycover} for a sketch. 

\begin{figure}[t]
\centering
\includegraphics{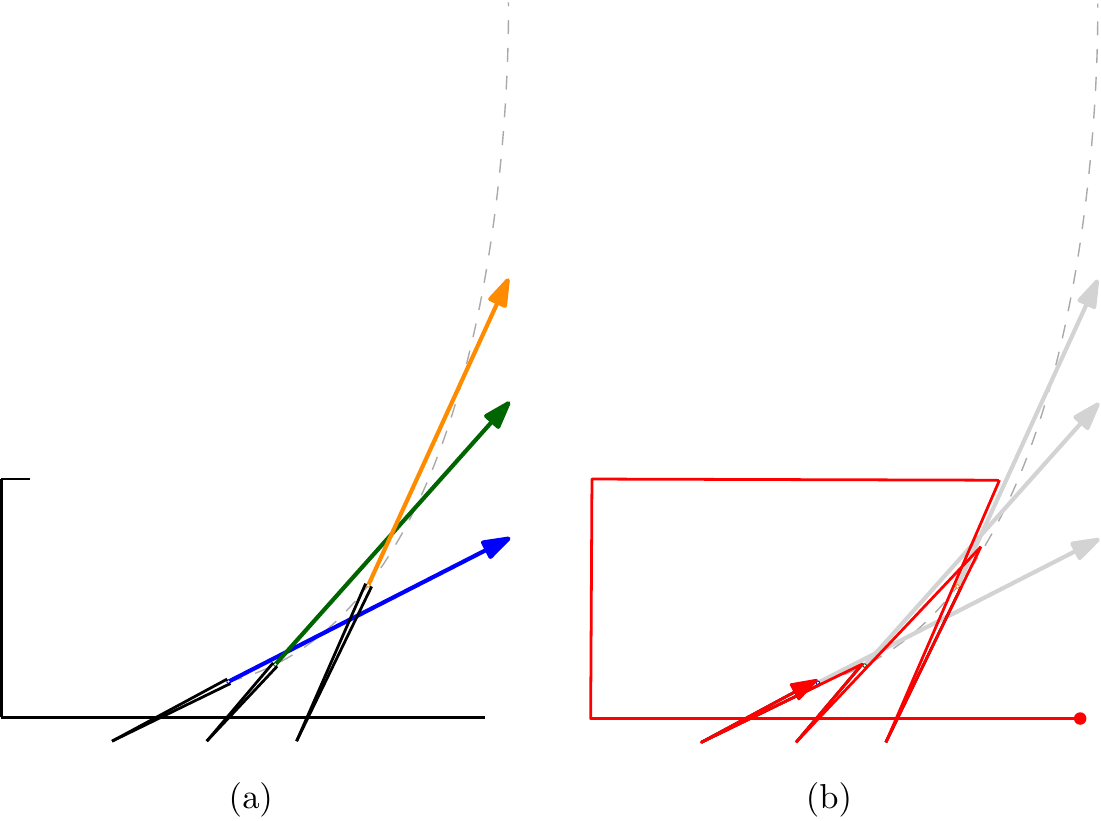}
\caption{(a) Sketch of a reduction of a circle graph with three chords. Segments shown in black. (b) Polyline of $2n+3$ links covering the constructed segments, corresponding to a Hamiltonian Path traversing the yellow, green, and blue ray in that order.}
\label{fig:consegpolycover}
\end{figure}

Now we have $2n+3$ segments in total, where $n$ is the number of chords in the original graph.
\begin {lemma}
We can cover all segments using a polyline of $2n+3$ links if and only if the circle graph has a Hamiltonian Path.
\end {lemma}

\begin {proof}
To see why this is true, consider that since none of the segments are collinear and no three segments intersect in the same point, a suitable polyline must fully cover one segment with each link. For a polyline to be able to bend from fully covering one segment to fully covering another, either the segments must have a shared endpoint, or the supporting lines of the segments must intersect in a point not contained in either segment.
Segment $s_h$ intersects all needle segments in their interior and is parallel to $s_t$, so we know that a suitable polyline must start\footnote{A suitable polyline could also end with $s_h$, but we will define the polyline to be in this direction for ease of notation} by covering $s_h$, and it must bend at the common endpoint with $s_v$ and then fully cover $s_v$. All of the intersection points between $s_v$ and the supporting lines of needle segments lie below the endpoint it shares with $s_h$, so to be able to cover $s_v$ with the second link the polyline must next connect to $s_t$, meaning it bends at the shared endpoint of $s_v$ and $s_t$. The third link is horizontal, covering $s_t$. Since the supporting line of $s_t$ intersects all of the rays, the polyline can bend to any needle segment for its next link. 

Since we have covered our additional segments $s_h$, $s_v$, and $s_t$, the rest of the $2n$ links must cover one needle segment each. Observe that the needle segments all extend downward to below the lowest intersection point between supporting lines. This means that when a link covers a needle segment when travelling downward, the next link must then travel upward on the other half of the needle, as all intersection points with supporting lines of other segments lie in the segment's interior. When the next link then covers a needle segment when travelling upward, the only places the polyline can viably bend next are near places where the ray associated with the previous needle intersects another ray. So we can cover the $2n$ needle segments using a polyline of $2n$ links if and only if there is a Hamiltonian Path in the ray intersection graph and thus a Hamiltonian Path in the original circle graph. 
\end {proof}

Since the transformation is polynomial, we know the problem is NP-hard. We can also see that the problem is in NP, since for any instance we can expect that if a polyline of $k$ links exists covering a set $L$, one must also exist where each vertex has coordinates of polynomial complexity, since the vertices could all lie on the intersection points of the supporting lines of the segments, or otherwise on points with rational coordinates on those supporting lines. This polyline could serve as a certificate for the verification algorithm. This gives the following theorem:
\begin{theorem}
  \consegcover is NP-complete.
\end{theorem}

\section {\curvesimp}

Finally, we reduce \consegcover to \curvesimp.
As a problem instance, we are given a set $L$ of $n$ non-collinear line segments in the plane whose union is connected. We construct an input polyline of polynomial size that completely covers the set of segments and no other points. We could do this, for example, by treating the segment endpoints and intersection points as vertices of a graph connected by edges, and have our polyline be the path of a breadth-first search through the graph. We set $\delta$ to $0$. 
Now we know that, since a simplification must cover the union of $L$, any simplification of our input polyline that has $n$ links must cover each segment in $L$ completely with one link. This means such a simplification would be a solution to our instance of \consegcover.
Since the reduction is polynomial in size, we know that this variant of the GCS problem is NP-hard, and using a similar argument to the one for \consegcover it is easy to see that it is in NP as well.
\begin{theorem}
\label{thm:finalcellhardsimplified}
  \curvesimp, restricted to instances where $\delta=0$, is NP-complete.
\end{theorem}

\subsection {Non-zero $\delta$}
\label{sec:nonzerodelta}
We can also extend this reduction to non-zero \(\delta\) by picking \(\delta > 0\) but still small enough such that it would not change the combinatorial structure of the space the polyline can lie in, so each link of the polyline must still correspond to exactly one segment in $L$.
For the segments we have constructed in the earlier reduction, we will show that setting \(\delta < \frac{3}{4n!}\) will guarantee the structure of the space will not change. So a simplified polyline of $2n+3$ links with \(0 < \delta < \frac{3}{4n!}\) will only exist if and only if it also exists for \(\delta = 0\).

For space reasons, we only sketch the ideas of the proof here; details can be found in Appendix~\ref{sec:nonzeroapp} in the full version of this section. 

The main idea is to choose $\delta$ sufficiently small so that there are no additional intersections between extensions of segments that are not supposed to intersect. When \(\delta = 0\), the space the simplified polyline can occupy is exactly the input segments, but for non-zero \(\delta\), the polyline does not have to exactly cover the original segment. If we center two circles with radius \(\delta\) on the endpoints of a segment, the two inner tangents of these circles will form the bounding lines of a cone that covers all possible polyline links that are able to ``cover'' a segment. We will call the part of the cone that is within \(\delta\) of the segment the \emph{tip} of the cone, and the rest of the cone the \emph{tail} of the cone. For \(\delta = 0\), the supporting line for the segment forms a degenerate cone of width 0. To preserve the combinatorial structure, fattening the cones cannot introduce intersections between cone tails, as these correspond to two segments' supporting lines intersecting in the exterior of the segments.
To simplify the algebra, we consider a slightly larger cone, between the lines connecting points created by going \(2\delta\) to the left and right of the original endpoints of the needle segments. It is easy to see that these larger cones contain the true cones.
We reach the bound on $\delta$ given above by case distinction of different configurations of potentially intersecting cones, and taking the minimum.

Since we can have a small enough \(\delta\) of polynomial bit complexity, this means the \textsc{Directed Curve Simplifiction} problem is NP-hard in general, as for larger values of \(\delta\) the construction could be scaled up. 

If the general problem is in NP is hard to say, since our approach for showing this for the previous problems does not extend, and it might be possible that inputs exist where the only possible simplifications of $k$ links have vertex coordinates of exponential bit complexity. This remains an open problem.

\begin{theorem}  
\label{thm:finalcellhard}
  \curvesimp is NP-hard.
\end{theorem}

\section {Conclusion}
We have shown that \curvesimp is NP-hard, which completes the results in Table~\ref {tab:results} and completely settles the complexity of global curve simplification under the Hausdorff distance.

As the main tool in our reduction, we have shown that every circle graph can be embedded as an intersection graph of rays with coordinates of polynomial complexity.
It is still an open question if it is possible to embed every circle graph as rays with coordinates of logarithmic complexity.
Whether \curvesimp is in NP is another open problem that remains.

\section* {Acknowledgements}
The authors would like to thank Birgit Vogtenhuber, Tillman Miltzow, and anonymous reviewers for valuable discussions and suggestions.
Mees van de Kerkhof and Maarten L\"offler are (partially) supported by the Dutch Research Council under grant  628.011.00. Maarten L\"offler is supported by the Dutch Research Council under grant 614.001.50.

\bibliographystyle{splncs04}
\bibliography{refs}

\section {Section~\ref{sec:nonzerodelta} (full)}
\label{sec:nonzeroapp}
We can also extend this reduction to non-zero \(\delta\) by picking \(\delta > 0\) but still small enough such that it would not change the combinatorial structure of the space the polyline can lie in, so each link of the polyline must still correspond to exactly one segment in $L$.
For the segments we have constructed in the earlier reduction, we will show that setting \(\delta < \frac{3}{4n!}\) will guarantee the structure of the space will not change. So a simplified polyline of $2n+3$ links with \(0 < \delta < \frac{3}{4n!}\) will only exist if and only if it also exists for \(\delta = 0\).

\subsubsection {Overview.}

The main idea is to choose $\delta$ sufficiently small so that there are no additional intersections between extensions of segments that are not supposed to intersect. When \(\delta = 0\), the space the simplified polyline can occupy is exactly the input segments, but for non-zero \(\delta\), the polyline does not have to exactly cover the original segment. If we center two circles with radius \(\delta\) on the endpoints of a segment, the two inner tangents of these circles will form the bounding lines of a cone that covers all possible polyline links that are able to ``cover'' a segment. We will call the part of the cone that is within \(\delta\) of the segment the \emph{tip} of the cone, and the rest of the cone the \emph{tail} of the cone. For \(\delta = 0\), the supporting line for the segment forms a degenerate cone of width 0. To preserve the combinatorial structure, fattening the cones cannot introduce intersections between cone tails, as these correspond to two segments' supporting lines intersecting in the exterior of the segments.
To simplify the algebra, we consider a slightly larger cone, between the lines connecting points created by going \(2\delta\) to the left and right of the original endpoints of the needle segments. It is easy to see that these larger cones contain the true cones.
We reach the bound on $\delta$ given above by case distinction of different configurations of potentially intersecting cones, and taking the minimum.
Since we can have a small enough \(\delta\) of polynomial bit complexity, this implies that the \textsc{Directed Curve Simplifiction} problem is NP-hard in general, as for larger values of \(\delta\) the construction could be scaled up. 

\subsubsection {Details.}

In the remainder of this section, we will treat $\delta$ as an unknown and aim to determine a value where we can guarantee any \(\delta\) smaller than that value will not change the combinatorial structure. 

First, consider that we have constructed $2n+3$ segments. For \(\delta = 0\), the space the polyline can lie in is exactly these segments, but for non-zero \(\delta\), this space is what you get when each segment gets replaced by the Minkowski sum of the original segment and a disk of radius \(\delta\). This means that the polyline does not have to exactly cover the original segment, meaning additional angles of approaching a segment open up. If we center two circles with radius \(\delta\) on the endpoints of a segment, the two inner tangents of these circles will form the bounding lines of a cone that covers all possible polyline links that are able to ``cover'' a segment. We will call the part of the cone that is within \(\delta\) of the segment the \emph{tip} of the cone, and the rest of the cone the \emph{tail} of the cone. For \(\delta = 0\), the supporting line for the segment forms a degenerate cone of width 0. To preserve the combinatorial structure, fattening the cones cannot introduce intersections between cone tails, as these correspond to two segments' supporting lines intersecting in the exterior of the segments.
 We do not have to consider the cones associated with the leading segments, since they already intersect all of the needles' cones in the tip. 
 So we will focus on the cones associated with needle segments.

Instead of taking the actual cones and the half-lines that bound them, we will consider a larger cone that is easier to calculate with. As bounding lines, we will use endpoints created by going \(2\delta\) to the left and right of the original endpoints of the needle segments. For a ray starting at point \((a,a!)\), the original endpoints for one of its needle segments are \((x^a_\ell,y_\ell)\) and a point very close to \(a,a!\), where \(x^a_\ell\) is the x-coordinate of the supporting line of the ray at \(y=y_\ell\). This gives the segment a slope of \(\frac{a! - y_\ell}{a-x^a_\ell}\). The slope of the left bounding line of the cone is then \(\frac{a! - y_\ell}{a-x^a_\ell - 4\delta}\). So, departing from the intersection point of the bounding line and the needle segment, for every \((a!-y_\ell)\) we move upwards, the bounding line lies \(4\delta\) further to the left as compared to the needle segment and/or the ray associated with the needle. 

To see how small \(\delta\) needs to be, we consider any two rays where one ray connects the points \((a,a!),(b,b!)\) and  one ray connects the points \((c,c!),(d,d!)\). We will refer to the cones induced by their needle segments as the \((a,b)\)-cone and \((c,d)\)-cone respectively. We assume \(a<c\). Now there are three cases again analogous to those in the proof for Lemma~\ref{lem:case1}, see Figure~\ref{fig:casedistinction2}:
\begin{figure}
      \centering
      \includegraphics[width=0.85\textwidth]{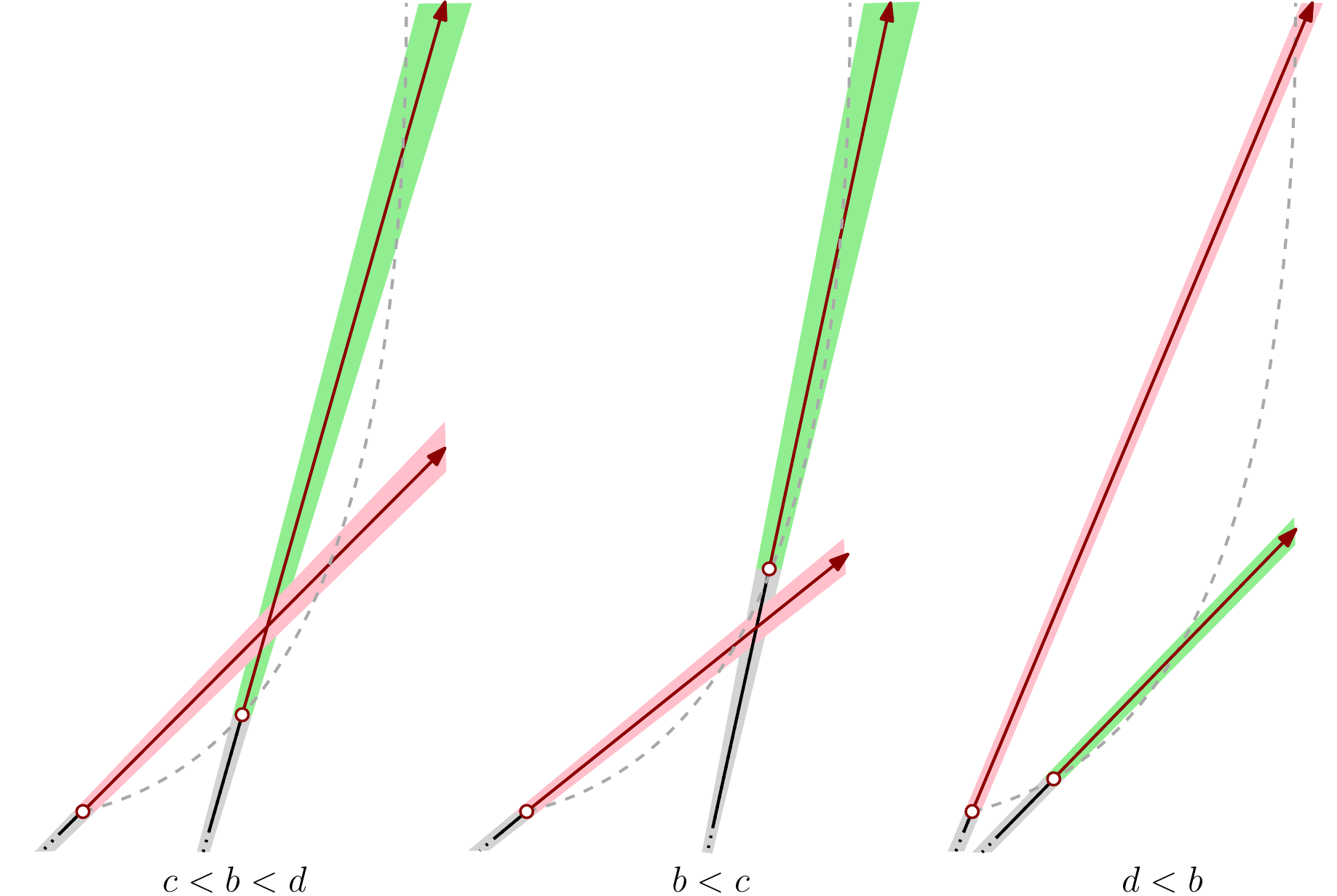}
      \caption{The three cases for two needle segments' cones. The segments are shown in black. The tips of the cones are shown in gray. The tail of the $(a,b)$-cone is shown in pink and the tail of the $(c,d)$-cone in gray.
      For the second and third case, the cone can be fattened (i.e. \(\delta\) increased) as long as this does not cause the tails to intersect.
      }
      \label{fig:casedistinction2}
  \end{figure}
\begin{itemize}
\item The case where \(c < b < d\). Here, the tails of the cones already intersect, so no \(\delta\) is going to introduce a new intersection.

\item The case where \(b < c\). Here, there is already an intersection between the tail of the \((a,b)\)-cone and the tip of the \((c,d)\)-cone. It is important that after fattening the cones there is no intersection that lies in both tails. At $y=c!$ we know that the right bounding line of the \((c,d)\)-cone has shifted \(2\delta\) to the right of where it would be for \(\delta = 0\). Likewise, the left bounding line of the \((a,b)\)-cone has shifted $4\delta \cdot \frac{c!-a!}{a! - y_\ell}$ leftwards. If the total shift is less than the original horizontal distance between the rays at this $y$-coordinate, we know the intersection between the cones still lies in the tip of the $(c,d)$-cone. The horizontal distance between the unshifted lines at this value for $y$, $dist_h$ can be written as \[dist_h = a + \frac{c!-a!}{b!-a!}\cdot(b-a) - c\]
For a fixed $b$, we get a smaller value the higher $a$ is, as this increases the slope of the $(a,b)$-cone. The value is also smaller the smaller $c$ is, as this decreases the vertical distance and therefore the horizontal distance to the point $(c,c!)$ as well. Substituting $a = b-1, c= b+1$ into the equation gives $dist_h = \frac{b^2}{b-1} -1$.
This is smallest when $b$ is minimal, so the minimum distance of $3$ is achieved when $b = 2$. This means that a bound on our \(\delta\) is \[3 > 4\delta \cdot \frac{c!-a!}{a! - y_\ell}\] We can replace the right hand side of this inequality by an upper bound \[3 > 4\delta \cdot n!\] This gives us an upper bound on the value for \(\delta\) of \(\delta < \frac{3}{4n!}\).
\item The case where  \(d < b\). Here there is an intersection of the cones in the tips. We need to make sure not to introduce new intersections of the tails. For this, the left bounding line of the \((c,d)\)-cone needs to intersect the right bounding line of the \((a,b)\)-cone in the tip. We know that at \(y = a!\), which is near where the tip of the \((a,b)\)-cone ends, the right bounding line has shifted exactly \(2\delta\) to the right compared to when \(\delta=0\). At \(y = c!\), the left bounding line of the \((c,d)\)-cone will have shifted \(2\delta\) to the left, so at \(y = a!\) we know that the total shift is less than \(4\delta\). The horizontal distance \(dist_h\) between the supporting lines of the two rays at \(y=a!\) is equal to \[dist_h = c - \frac{c!-a!}{d!-c!}\cdot(d-c) - a\] For a fixed \(c\), the smallest distance is obtained when \(a = c-1\) (as increasing the value of \(a\) decreases the vertical distance between \(c\) and \(a\) meaning they are also closer horizontally at \(y = a!\)), and \(d = c+1\) (as decreasing \(d\) decreases the slope of the \((c,d)\)-cone, and so the supporting line will be further left at \(y=a!\)). Plugging these values into the equation for \(dist_h\) lets us rewrite it to \[dist_h = 1 - \frac{c-1}{c^2}\] Since \(c\ge2\) in this case, we know the smallest possible distance is \(\frac{3}{4}\). Therefore, we know this case will pose no problems if \(4\delta < \frac{3}{4} \rightarrow \delta < \frac{3}{16}\), as the shift is smaller than the original distance so the intersection of the cones still lies in the tips.
\end{itemize}

Over all three cases, the bound \(\delta < \frac{3}{4n!}\) is smallest and so it is the one that should be used. 
This upper bound is not necessarily tight, but we do now know that we can find non-zero values for \(\delta\) that can be represented in a polynomial number of bits that do not change the combinatorial structure of the input for the simplification problem. So a simplified polyline of $2n+3$ links with \(0 < \delta < \frac{3}{4n!}\) will only exist if and only if it also exists for \(\delta = 0\).
Since we can have a small enough \(\delta\) of polynomial bit complexity, this means the \textsc{Directed Curve Simplifiction} problem is hard in general, as for larger values of \(\delta\) the construction could be scaled up.  

If the general problem is in NP is hard to say, since our approach for showing this for the previous two problems does not extend, and it might be possible that inputs exist where the only possible simplified polylines of $k$ links have vertex coordinates of exponential bit complexity. We leave this as an open problem for now.


\begin{theorem}  
\label{thm:finalcellhard2}
  \curvesimp is NP-hard.
\end{theorem}

Whether \curvesimp is also in NP is hard to say, since our approach for showing this for the previous problems does not extend, and it might be possible that inputs exist where the only possible simplifications of $k$ links have vertex coordinates of exponential bit complexity. This remains an open problem.

\end {document}